\documentclass[conference,letterpaper]{IEEEtran}
\IEEEoverridecommandlockouts
\usepackage{amsmath,amssymb,amsfonts}
\usepackage[ruled,vlined]{algorithm2e}
\usepackage{graphicx}
\usepackage{textcomp}
\usepackage{xcolor}
\usepackage[utf8]{inputenc}
\usepackage{marvosym}
\usepackage{amsthm}
\usepackage{float}
\usepackage{url}
\usepackage{caption}
\usepackage{enumitem}

\allowdisplaybreaks

\newtheorem{definition}{Definition}
\newtheorem{theorem}{Theorem}
\newtheorem{problem}{Problem}
\newtheorem{proposition}{Proposition}
\newtheorem{lemma}{Lemma}
\newtheorem{remark}{Remark}

\newcommand\recht\operatorname
\usepackage[labelformat=simple]{subcaption}

\usepackage{tikz}
\usetikzlibrary{decorations.pathmorphing}
\usetikzlibrary{shapes.multipart}
\usetikzlibrary{arrows.meta}

\usepackage{biblatex}
\title{
\fontsize{14pt}{16pt}\selectfont\textbf{Data Sanitisation Protocols for the Privacy Funnel with Differential Privacy Guarantees}}

\author{
\IEEEauthorblockN{Milan Lopuhaä-Zwakenberg, Haochen Tong, and Boris \v{S}kori\'{c}}
\IEEEauthorblockA{Department of Mathematics and Computer Science\\
Eindhoven University of Technology\\
Eindhoven, the Netherlands\\
email: \{m.a.lopuhaa,b.skoric\}@tue.nl, h.tong@student.tue.nl}
}

\begin{document}

\IEEEtitleabstractindextext{%
    \begin{abstract}
In the Open Data approach, governments and other public organisations want to share their datasets with the public, for accountability and to support participation. Data must be opened in such a way that individual privacy is safeguarded. The Privacy Funnel is a mathematical approach that produces a sanitised database that does not leak private data beyond a chosen threshold. The downsides to this approach are that it does not give worst-case privacy guarantees, and that finding optimal sanitisation protocols can be computationally prohibitive. We tackle these problems by using differential privacy metrics, and by considering local protocols which operate on one entry at a time. We show that under both the Local Differential Privacy and Local Information Privacy leakage metrics, one can efficiently obtain optimal protocols. Furthermore, Local Information Privacy is both  more closely aligned to the privacy requirements of the Privacy Funnel scenario, and more efficiently computable. We also consider the scenario where each user has multiple attributes, for which we define \emph{Side-channel Resistant Local Information Privacy}, and we give efficient methods to find protocols satisfying this criterion while still offering good utility. Finally, we introduce \emph{Conditional Reporting}, an explicit LIP protocol that can be used when the optimal protocol is infeasible to compute, and we test this protocol on real-world and synthetic data. Experiments on real-world and synthetic data confirm the validity of these methods.
    \end{abstract}
}

\maketitle

\thispagestyle{plain}
\pagestyle{plain}

\IEEEdisplaynontitleabstractindextext


\section{Introduction}

This paper is an extended version of \cite{lopuhaa2020privacy}. Under the Open Data paradigm, governments and other public organisations want to share their collected data with the general public. This increases a government's transparency, and it also gives citizens and businesses the means to participate in decision-making, as well as using the data for their own purposes. However, while the released data should be as faithful to the raw data as possible, individual citizens' private data should not be compromised by such data publication.

Let $\mathcal{X}$ be a finite set. Consider a database $\vec{X} = (X_1,\ldots,X_n) \in \mathcal{X}^n$ owned by a data aggregator, containing a data item $X_i \in \mathcal{X}$ for each user $i$ (For typical database settings, each user's data is a vector of attributes $X_i = (X_i^1,\ldots,X_i^m)$; we will consider this in more detail in Section \ref{sec:mult}). This data may not be considered sensitive by itself, but it might be correlated to a secret $S_i$. 
For instance, $X_i$ might contain the age, sex, weight, skin colour, and average blood pressure of person $i$, while $S_i$ is the presence of some medical condition.
To publish the data without leaking the $S_i$, the aggregator releases a privatised database $\vec{Y} = (Y_1,\ldots,Y_n)$, obtained from applying a sanitisation mechanism $\mathcal{R}$ to $\vec{X}$. One way to formulate this is by considering the \emph{Privacy Funnel}:

\begin{problem} \label{prob:pf} \emph{(Privacy Funnel, \cite{calmon2017principal})}
Suppose the joint probability distribution of $\vec{S}$ and $\vec{X}$ is known to the aggregator, and let $M \in \mathbb{R}_{\geq 0}$. Then, find the sanitisation mechanism $\mathcal{R}$ such that $\recht{I}(\vec{X};\vec{Y})$ is maximised while $\recht{I}(\vec{S};\vec{Y}) \leq M$.
\end{problem}

There are two difficulties with this approach:
\begin{enumerate}
\item Finding and implementing good privatization mechanisms that operate on all of $\vec{X}$ can be computationally prohibitive for large $n$, as the complexity is exponential in $n$ \cite{ding2019submodularity}\cite{prasser2014arx}.
\item Taking mutual information as a leakage measure has as a disadvantage that it gives guarantees about the leakage in the average case. If $n$ is large, this still leaves room for the sanitisation protocol to leak undesirably much information about a few unlucky users.
\end{enumerate}

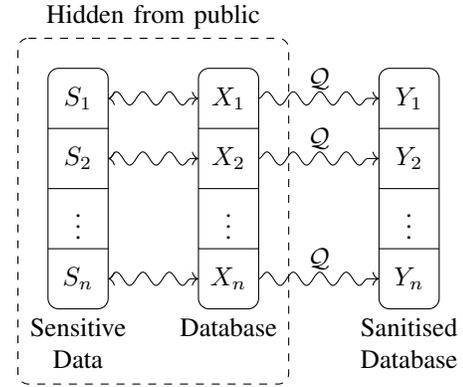
\begin{figure}
\centering
\begin{tikzpicture}[scale = 0.4, every text node part/.style={align=center}]
\draw[rounded corners] (1,1) -- (3,1) -- (3,-7) --node[below]{Sensitive \\ Data} (1,-7) -- cycle;
\draw[-] (1,-1)--(3,-1);
\draw[-] (1,-3)--(3,-3);
\draw[-] (1,-5)--(3,-5);
\draw (2,0) node{$S_1$};
\draw (2,-2) node{$S_2$};
\draw (2,-4) node{$\vdots$};
\draw (2,-6) node{$S_n$};
\draw[rounded corners] (6,1) -- (8,1) -- (8,-7) --node[below]{Database} (6,-7) -- cycle;
\draw[-] (6,-1)--(8,-1);
\draw[-] (6,-3)--(8,-3);
\draw[-] (6,-5)--(8,-5);
\draw (7,0) node{$X_1$};
\draw (7,-2) node{$X_2$};
\draw (7,-4) node{$\vdots$};
\draw (7,-6) node{$X_n$};
\draw[rounded corners] (12,1) -- (14,1) -- (14,-7) --node[below, text width = 1.4cm]{Sanitised Database} (12,-7) -- cycle;
\draw[-] (12,-1)--(14,-1);
\draw[-] (12,-3)--(14,-3);
\draw[-] (12,-5)--(14,-5);
\draw (13,0) node{$Y_1$};
\draw (13,-2) node{$Y_2$};
\draw (13,-4) node{$\vdots$};
\draw (13,-6) node{$Y_n$};
\draw[->,decorate,decoration={snake}](8,0) -- node[above]{$\mathcal{Q}$} (12,0);
\draw[->,decorate,decoration={snake}](8,-2) -- node[above]{$\mathcal{Q}$} (12,-2);
\draw[->,decorate,decoration={snake}](8,-6) -- node[above]{$\mathcal{Q}$} (12,-6);
\draw[<->,decorate,decoration={snake}] (3,0) -- (6,0);
\draw[<->,decorate,decoration={snake}] (3,-2) -- (6,-2);
\draw[<->,decorate,decoration={snake}] (3,-6) -- (6,-6);
\draw[rounded corners,dashed] (0,2) --node[above]{Hidden from public}  (9,2) -- (9,-9.5) -- (0,-9.5) -- cycle;
\end{tikzpicture}
\captionsetup{font={footnotesize,rm},justification=centering,labelsep=period}
\caption{\it Model of the Privacy Funnel with local protocols.} 
\label{fig:sit}
\end{figure}

To deal with these two difficulties, we make two changes to the general approach. First, we look at \emph{local} data sanitisation, i.e., we consider optimization protocols $\mathcal{Q}\colon \mathcal{X} \rightarrow \mathcal{Y}$, for some finite set $\mathcal{Y}$, and we apply $\mathcal{Q}$ to each $X_i$ individually; this situation is depicted in Figure \ref{fig:sit}. 
Local sanitisation can be implemented efficiently.
In fact, this approach is often taken in the Privacy Funnel setting \cite{makhdoumi2014information}\cite{ding2019submodularity}. Second, to ensure strong privacy guarantees even in worst-case scenarios, we take stricter notions of privacy, based on Local Differential Privacy (LDP) \cite{kasiviswanathan2011can}. For these metrics, we develop methods to find optimal protocols. Furthermore, for situations where the optimal protocol is computationally unfeasible to find, we introduce a new protocol, {\em Conditional Reporting} (CR), that takes advantage of the fact that only $S_i$ needs to be protected. Determining CR only requires finding the root of a onedimensional increasing function, which can be done fast numerically.

\subsection{New contributions}

In this paper, we adapt two Differential Privacy-like privacy metrics to the Privacy Funnel situation, namely Local Differential Privacy (LDP) \cite{kasiviswanathan2011can} and Local Information Privacy (LIP) \cite{jiang2019local}\cite{salamatian2020privacy}. 
We modify these metrics so that they measure leakage about the underlying $S$ rather than $X$ itself (for notational convenience, we write $S,X,Y$ rather than $S_i,X_i,Y_i$ throughout the rest of this paper). For a given level of leakage, we are interested in the privacy protocol that maximises the mutual information between input $X_i$ and output $Y_i$. 
Adapting methods from \cite{kairouz2014extremal} on LDP and \cite{rassouli2018perfect} on perfect privacy, we prove the following Theorem:

\begin{theorem}[Theorems \ref{thm:ldp} and \ref{thm:lip} paraphrased]
Suppose $a = \#\mathcal{X},c = \#\mathcal{S}$, and $\recht{p}_{X,S}$, as well as a privacy level $\varepsilon \geq 0$ are given.

\begin{enumerate}
    \item  The optimal $\varepsilon$-LDP protocol can be found by enumerating the vertices of a polytope in $a^2-a$ dimensions defined by $a(c^2-c)$ inequalities.
    \item The optimal $\varepsilon$-LIP protocol can be found by enumerating the vertices of a polytope in $a-1$ dimensions defined by $2ac$ inequalities.
\end{enumerate}
\end{theorem}

The descriptions of these polytopes, and how they relate to the optimisation problem, are discussed in Sections \ref{sec:ldp} and \ref{sec:lip}, respectively. 
Since the complexity of the polytope vertex enumeration depends significantly on both its dimension and the number of defining inequalities \cite{barany2000polytopes}, finding optimal LIP protocols can be done significantly faster than finding optimal LDP protocols. 
Furthermore, we will argue that LIP is a privacy metric that more accurately captures information leakage than LDP in the Privacy Funnel scenario. For these two reasons we only consider LIP in the remainder of the paper, although many results can also be formulated for LDP.

A common scenario is that a user's data $X$ consists of multiple attributes, i.e. $X = (X^1,\ldots,X^m)$. Here one can consider an attacker model where the attacker has access to some of the $X^j$. In this situation $\varepsilon$-LIP does not accurately reflect a user's privacy. Because of this, we introduce a new privacy metric called \emph{Sidechannel-Resistant LIP} that takes such sidechannels into account. We expand the vertex enumeration methods outlined above to find optimal SRLIP methods in Section \ref{sec:mult}.

Finding the optimal protocols can become computationally unfeasible for large $a$ and $c$. In such a situation, one needs to resort to explicitely given protocols. In the literature there is a wealth of protocols that satisfy $\varepsilon$-LDP w.r.t. $X$. These certainly work in our situation, but they might not be ideal, because these are designed to obfuscate all information about $X$, rather than just the part that relates to $S$. For this reason, we introduce Conditional Reporting (CR), a privacy protocol that focuses on hiding $S$ rather than $X$, in Section \ref{sec:cr}. Finding the appropriate CR protocol for a given probability distribution and privacy level can be done fast numerically.

In Section \ref{sec:exp}, we test the methods and protocols discussed above on both synthetic and real data. 
Compared to \cite{lopuhaa2020privacy},
new content in this extended paper are Section \ref{sec:cr}, the experiments on real data, and the extended literature review.

\subsection{Related work}

The Privacy Funnel (PF) setting was introduced in \cite{makhdoumi2014information}, to provide a framework for obfuscating data in such a way that the obfuscated data remains as faithful as possible to the original, while ensuring that the information leakage about a latent variable is limited. The Privacy Funnel is related to the Information Bottleneck (IB) \cite{tishby2000information}, a problem from machine learning that seeks to compress data as much as possible, while retaining a minimal threshold of information about a latent variable. 
In PF as well as IB, both utility and leakage are measured via mutual information. Many approaches to finding the optimal protocols in PF also work for IB and vice versa \cite{kung2018compressive}\cite{ding2019submodularity}. A wider range of privacy metrics for the Privacy Funnel, and their relation to Differential Privacy, is discussed in \cite{salamatian2020privacy}.

Local Differential Privacy (LDP) was introduced in \cite{kasiviswanathan2011can}. 
It is an adaptation of Differential Privacy \cite{dwork2006calibrating} to a setting where there is no trusted central party to obfuscate the data. 
As a privacy metric, it has the advantage that it offers a privacy guarantee in any case, not just the average case, and that it does not depend on the data distribution. On the downside, it can be difficult to fulfill such a stringent definition of privacy, and many relaxations of (Local) Differential Privacy have been proposed \cite{cuff2016differential}\cite{dwork2006our}\cite{dwork2016concentrated}\cite{mironov2017renyi}. 
We are particularly interested in Local Information Privacy (LIP) \cite{jiang2019local}\cite{salamatian2020privacy}, also called Removal Local Differential Privacy \cite{erlingsson2020encode}. 
LIP retains the worst-case guarantees of LDP, but is less restrictive, and can 
take advantage of a known distribution. In the context where only part of the data is considered secret, many privacy metrics fall under the umbrella of Pufferfish Privacy \cite{kifer2014pufferfish}.

In \cite{kairouz2014extremal}, a method
was introduced for finding optimal LDP-protocols for a wide variety of utility metrics, including mutual information. 
The method relies on finding the vertices of a polytope, but since this is the well-studied Differential Privacy polytope, its vertices can be described explicitly \cite{holohan2017extreme}. 
Similarly, \cite{rassouli2018perfect} uses a vertex enumeration method to find the optimal protocol in the perfect privacy situation, i.e. when the released data is independent of the secret data. The complexity of vertex enumeration is discussed in \cite{avis1992pivoting}\cite{barany2000polytopes}.

\section{Mathematical Setting} \label{sec:maths}

The database $\vec{X} = (X_1,\ldots,X_n)$ consists of a data item $X_i$ for each user $i$, each an element of a given finite set $\mathcal{X}$. Furthermore, each user has sensitive data $S_i \in \mathcal{S}$, which is correlated with $X_i$; again we assume $\mathcal{S}$ to be finite (see Figure \ref{fig:sit}). 
We assume that each $(S_i,X_i)$ is drawn independently from the same distribution $\recht{p}_{S,X}$ on $\mathcal{S} \times \mathcal{X}$ which is known to the aggregator through observing $(\vec{S},\vec{X})$ (if one allows for non-independent $X_i$, then differential privacy is no longer an adequate privacy metric \cite{cuff2016differential}{}\cite{salamatian2020privacy}). The aggregator, who has access to $\vec{X}$, sanitises the database by applying a sanitisation protocol (i.e., a random function) $\mathcal{Q}\colon \mathcal{X} \rightarrow\mathcal{Y}$ to each $X_i$, outputting $\vec{Y} = (Y_1,\ldots,Y_n) = (\mathcal{Q}(X_1),\ldots,\mathcal{Q}(X_n))$. The aggregator's goal is to find a $\mathcal{Q}$ that maximises the information about $X_i$ preserved in $Y_i$ (measured as $\recht{I}(X_i;Y_i)$) while leaking only minimal information about $S_i$.

Without loss of generality we write $\mathcal{X} = \{1,\ldots,a\}$, $\mathcal{Y} = \{1,\ldots,b\}$ and $\mathcal{S} = \{1,\ldots,c\}$ for integers $a,b,c$. We omit the subscript $i$ from $X_i$, $Y_i$, $S_i$ as no probabilities depend on it, and we write such probabilities as $\recht{p}_x$, $\recht{p}_s$, $\recht{p}_{x|s}$, etc., which form vectors $\recht{p}_{X}$, $\recht{p}_{S|x}$, etc., and matrices $\recht{p}_{X|S}$, etc.

As noted before, instead of looking at the mutual information $\recht{I}(S;Y)$, we consider two different, related measures of sensitive information leakage known from the literature. The first one is an adaptation of LDP, the \emph{de facto} standard in information privacy \cite{kasiviswanathan2011can}:

\begin{definition} \emph{($\varepsilon$-LDP)}
Let $\varepsilon \in \mathbb{R}_{\geq 0}$. We say that $\mathcal{Q}$ satisfies $\varepsilon$-LDP w.r.t. $S$ if 
\begin{equation}
\forall_{y\in{\cal Y}}\forall_{s,s'\in\cal S} \quad
\frac{\mathbb{P}(Y = y | S = s)}{\mathbb{P}(Y = y  | S = s')} \leq \textrm{\emph{e}}^{\varepsilon}.
\end{equation}
\end{definition}

Most literature on LDP considers LDP w.r.t. $X$, i.e. $\frac{\mathbb{P}(Y = y | X = x)}{\mathbb{P}(Y = y  | X = x')} \leq \textrm{e}^{\varepsilon}$ for all $x,x',y$. Throughout this paper, by $\varepsilon$-LDP we always mean $\varepsilon$-LDP w.r.t. $S$, unless otherwise specified.

The LDP metric reflects the fact that we are only interested in hiding sensitive data, rather than all data; it is a specific case of what has been named `pufferfish privacy' \cite{kifer2014pufferfish}. The advantage of LDP compared to mutual information is that it gives privacy guarantees for the worst case, not just the average case. This is desirable in the database setting, as a worst-case metric guarantees the security of the private data of all users, while average-case metrics are only concerned with the average user. Another useful privacy metric is  \emph{Local Information Privacy} (LIP) \cite{jiang2019local}\cite{salamatian2020privacy}, also called Removal Local Differential Privacy \cite{erlingsson2020encode}:

\begin{definition} \emph{($\varepsilon$-LIP)}
Let $\varepsilon \in \mathbb{R}_{\geq 0}$. We say that $\mathcal{Q}$ satisfies $\varepsilon$-LIP w.r.t. $S$ if 
\begin{equation}
\forall_{y\in{\cal Y},s\in{\cal S}}\quad\quad
\textrm{\emph{e}}^{-\varepsilon} \leq \frac{\mathbb{P}(Y = y | S = s)}{\mathbb{P}(Y = y)} \leq \textrm{\emph{e}}^{\varepsilon}.
\end{equation}
\end{definition}

Compared to LDP, the disadvantage of LIP is that it depends on the distribution of $S$; 
this is not a problem in our scenario, as the aggregator, who chooses $\mathcal{Q}$, has access to the distribution of $S$. 
The advantage of LIP is that is more closely related to an attacker's capabilities: since $\frac{\mathbb{P}(Y = y | S = s)}{\mathbb{P}(Y = y)} = \frac{\mathbb{P}(S = s | Y = y)}{\mathbb{P}(S = s)}$, satisfying $\varepsilon$-LIP means that an attacker's posterior distribution of $S$ given $Y = y$ does not deviate from their prior distribution by more than a factor $\textrm{e}^{\varepsilon}$. The following lemma outlines the relations between LDP, LIP and mutual information (see Figure \ref{fig:rel}).

\begin{lemma} \emph{(See \cite{salamatian2020privacy})}
Let $\mathcal{Q}$ be a sanitisation protocol, and let $\varepsilon \in \mathbb{R}_{\geq 0}$.
\begin{enumerate}
\item If $\mathcal{Q}$ satisfies $\varepsilon$-LDP, then it satisfies $\varepsilon$-LIP.
\item If $\mathcal{Q}$ satisfies $\varepsilon$-LIP, then it satisfies $2\varepsilon$-LDP, and $\recht{I}(S;Y) \leq \varepsilon$.
\end{enumerate}
\end{lemma}

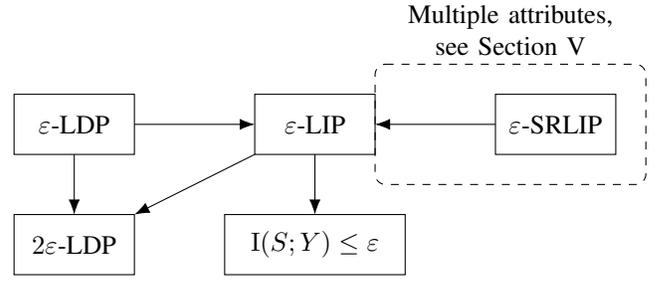
\begin{figure}
\centering
\begin{tikzpicture}[scale = 0.4, every text node part/.style={align=center}]
\draw[-] (0,0) -- (4,0) -- (4,-2) -- (0,-2) -- cycle;
\draw (2,-1) node{$\varepsilon$-LDP};
\draw[-] (0,-4) -- (4,-4) -- (4,-6) -- (0,-6) -- cycle;
\draw (2,-5) node{$2\varepsilon$-LDP};
\draw[-] (8,0) -- (12,0) -- (12,-2) -- (8,-2) -- cycle;
\draw (10,-1) node{$\varepsilon$-LIP};
\draw[-] (7,-4) -- (13,-4) -- (13,-6) -- (7,-6) -- cycle;
\draw (10,-5) node{$\recht{I}(S;Y) \leq \varepsilon$};
\draw[-] (16,0) -- (20,0) -- (20,-2) -- (16,-2) -- cycle;
\draw (18,-1) node{$\varepsilon$-SRLIP};
\draw[-{Latex[length=2mm]}] (4,-1) -- (8,-1);
\draw[-{Latex[length=2mm]}] (2,-2) -- (2,-4);
\draw[-{Latex[length=2mm]}] (8,-2) -- (4,-4);
\draw[-{Latex[length=2mm]}] (10,-2) -- (10,-4);
\draw[-{Latex[length=2mm]}] (16,-1) -- (12,-1);
\draw[-,dashed,rounded corners] (12,1) --node[above]{Multiple attributes, \\ see Section \ref{sec:mult}} (21,1) -- (21,-3) -- (12,-3) -- cycle;
\end{tikzpicture}
\captionsetup{font={footnotesize,rm},justification=centering,labelsep=period}
\caption{\it Relations between privacy notions. The multiple attributes setting is discussed in Section \ref{sec:mult}.} 
\label{fig:rel}
\end{figure}

\begin{remark} \label{rem:prior}
One gets robust equivalents of LDP and LIP by demanding that $\mathcal{Q}$ satisfy $\varepsilon$-LIP ($\varepsilon$-LDP) for a 
\underline{set} of distributions $\recht{p}_{S,X}$, instead of only 
a single distribution \cite{kifer2014pufferfish}. 
Letting $\recht{p}_{S,X}$ range over all possible distributions on $\mathcal{S} \times \mathcal{X}$ yields LIP (LDP) w.r.t. $X$.
\end{remark}

In this notation, instead of Problem \ref{prob:pf} we consider the following problem:

\begin{problem} \label{prob:pf2}
Suppose $\recht{p}_{S,X}$ is known to the aggregator, and let $\varepsilon \in \mathbb{R}_{\geq 0}$. Then, find the sanitisation protocol $\mathcal{Q}$ such that $\recht{I}(X;Y)$ is maximised while $\mathcal{Q}$ satisfies $\varepsilon$-LDP ($\varepsilon$-LIP, respectively) with respect to $S$.
\end{problem}

Note that this problem does not depend on the number of users $n$, and as such this approach will find solutions that are scalable w.r.t. $n$.

\section{Optimizing $\mathcal{Q}$ for $\varepsilon$-LDP} \label{sec:ldp}

Our goal is now to find the optimal $\mathcal{Q}$, i.e., the protocol that maximises $\recht{I}(X;Y)$ while satisfying $\varepsilon$-LDP, for a given $\varepsilon$. We can represent any sanitisation protocol as a matrix $Q \in \mathbb{R}^{b \times a}$, where $Q_{y|x} = \mathbb{P}(Y = y | X = x)$. 
Then, $\varepsilon$-LDP is satisfied if and only if
\begin{align}
\forall x\colon& \ \sum_y Q_{y|x} = 1,\\
\forall x,y\colon& \ 0 \leq Q_{y|x},\\
\forall s,s',y\colon& \ (Q\recht{p}_{X|s})_y \leq \textrm{e}^{\varepsilon} (Q\recht{p}_{X|s'})_y. \label{eq:ldp}
\end{align}
As such, for a given $\mathcal{Y}$, the set of $\varepsilon$-LDP-satisfying sanitisation protocols can be considered a closed, bounded, convex polytope $\Gamma$ in $\mathbb{R}^{b \times a}$. This fact allows us to efficiently find optimal protocols.

\begin{theorem} \label{thm:ldp}
Let $\varepsilon \in \mathbb{R}_{\geq 0}$. Let $\mathcal{Q}\colon \mathcal{X} \rightarrow \mathcal{Y}$ be a $\varepsilon$-LDP protocol that maximises $\recht{I}(X;Y)$, i.e., the protocol that solves Problem \ref{prob:pf2} w.r.t. LDP.
\begin{enumerate}
\item One can take $b = a$.
\item Let $\Gamma$ be the polytope described above, for $b = a$. Then the optimal $\mathcal{Q}$ corresponds to one of the vertices of $\Gamma$.
\end{enumerate}
\end{theorem}

\begin{proof}
The first result is obtained by generalising the results of \cite{kairouz2014extremal}: there this is proven for regular $\varepsilon$-LDP (i.e., w.r.t. $X$), but the arguments given in that proof hold just as well in our situation; the only difference is that their polytope is defined by the $\varepsilon$-LDP conditions w.r.t. $X$, but this has no impact on the proof. The second statement follows from the fact that $\recht{I}(X;Y)$ is a convex function in $\mathcal{Q}$; therefore its maximum on a bounded polytope is attained in one of the vertices.
\end{proof}

This theorem reduces the search for the optimal LDP protocol to enumerating the set of vertices of $\Gamma$, a $a(a-1)$-dimensional convex polytope.

One might argue that, since the optimal $\mathcal{Q}$ depends on $\recht{p}_{S,X}$, the publication of $\mathcal{Q}$ might provide an aggregator with information about the distribution of $S$. However, information on the distribution (as opposed to information of individual users' data) is not considered sensitive \cite{lopuhaa2019information}. In fact, the reason why the aggregator sanitises the data is because an attacker is assumed to have knowledge about this correlation, and revealing too much information about $X$ would cause the aggregator to use this information to infer information about~$S$.


\section{Optimizing $\mathcal{Q}$ for $\varepsilon$-LIP} \label{sec:lip}

If one uses $\varepsilon$-LIP as a privacy metric, one can find the optimal sanitisation protocol in a similar fashion. To do this, we again describe $\mathcal{Q}$ as a matrix, but this time a different one. Let $q \in \mathbb{R}^{b}$ be the probability mass function of $Y$, and let $R \in \mathbb{R}^{a \times b}$ be given by $R_{x|y} = \mathbb{P}(X = x | Y = y)$; we denote its $y$-th row by $R_{X|y} \in \mathbb{R}^a$. Then, a pair $(R,q)$ defines a sanitisation protocol $\mathcal{Q}$ satisfying $\varepsilon$-LIP if and only if
\begin{align}
\forall y \colon& \ 0 \leq q_y,\\
& \ Rq = \recht{p}_X, \label{eq:LIP2}\\
\forall y\colon& \ \sum_x R_{x|y} = 1,\label{eq:LIP3}\\
\forall x,y\colon& \ 0 \leq R_{x|y},\label{eq:LIP4}\\
\forall y,s\colon& \ \textrm{e}^{-\varepsilon}\recht{p}_s \leq \recht{p}_{s|X}R_{X| y} \leq \textrm{e}^{\varepsilon}\recht{p}_s. \label{eq:LIP5}
\end{align}

Note that (\ref{eq:LIP5}) defines the $\varepsilon$-LIP condition, since for a given $s,y$ we have $\frac{\recht{p}_{s|X}R_{X| y}}{\recht{p}_S} = \frac{\mathbb{P}(S=s|Y=y)}{\mathbb{P}(S=s)} = \frac{\mathbb{P}(Y=y|S=s)}{\mathbb{P}(Y=y)}$. (In)equalities (\ref{eq:LIP3}--\ref{eq:LIP5}) can be expressed as saying that for every $y \in \mathcal{Y}$ one has that $R_{X|y} \in \Delta$, where $\Delta$ is the convex closed bounded polytope in $\mathbb{R}^{\mathcal{X}}$ given by
\begin{equation}
\Delta = \left\{v \in \mathbb{R}^{\mathcal{X}}:  \begin{array}{l}
                          \sum_x v_x = 1,\\
                          \forall x: 0 \leq v_x,\\
                          \forall s: \textrm{e}^{-\varepsilon}\recht{p}_s \leq \recht{p}_{s|X}v \leq \textrm{e}^{\varepsilon}\recht{p}_s
                          \end{array}\right\}. \label{eq:delta}
\end{equation}
As in Theorem \ref{thm:ldp}, we can use this polytope to find optimal protocols:

\begin{theorem} \label{thm:lip}
Let $\varepsilon \in \mathbb{R}_{\geq 0}$, and let $\Delta$ be the polytope above. Let $\mathcal{V} = \{v_1,\ldots,v_M\}$ be its set of vertices. For $v_i \in \mathcal{V}$, let $\recht{H}(v_i)$ be its entropy, i.e.
\begin{equation}
    \recht{H}(v_i) = -\sum_{x \in \mathcal{X}} v_x \ln(v_{i,x}).
\end{equation}
Let $\hat{\alpha}$ be the solution to the optimisation problem
\begin{align}
\textrm{\emph{minimise}}_{\alpha \in \mathbb{R}^{M}} & \ \ \sum_{i=1}^M \recht{H}(v_i)\alpha_i \\
\textrm{\emph{subject}} \ \textrm{\emph{to}} & \ \ \forall i: \alpha_i \geq 0,\\
& \ \ \sum_{i=1}^M \alpha_iv_i = \recht{p}_X.
\end{align}
Then the $\varepsilon$-LIP protocol $\mathcal{Q}\colon \mathcal{X} \rightarrow \mathcal{Y}$ that maximises $\recht{I}(X;Y)$ is given by
\begin{align}
    \mathcal{Y} &= \{i \leq M: \hat{\alpha}_i > 0\},\\
    q_i &= \hat{\alpha}_i,\\
    R_{x|i} &= v_{i,x},
\end{align}
for all $i \in \mathcal{Y} \subseteq\{1,\ldots,M\}$ and all $x \in \mathcal{X}$. One has $b \leq a$.
\end{theorem}

\begin{proof}
This was proven for $\varepsilon = 0$ (i.e., when $S$ and $Y$ are independent) in \cite{rassouli2018perfect}, but the proof works similarly for $\varepsilon > 0$; the main difference is that the equality constraints of their (10) will be replaced by the inequality constraints of our (\ref{eq:LIP5}), but this has no impact on the proof presented there. 
\end{proof}

Since linear optimization problems can be solved fast, again the optimization problem reduces to finding the vertices of a polytope. 
The advantage of using LIP instead of LDP is that $\Delta$ is a $(a-1)$-dimensional polytope, while $\Gamma$ of Section \ref{sec:ldp} is $a(a-1)$-dimensional. The time complexity of vertex enumeration is linear in the number of vertices \cite{avis1992pivoting}, while the number of vertices can grow exponentially in the dimension of the polyhedron \cite{barany2000polytopes}. Together, this means that the dimension plays a huge role in the time complexity, hence we expect finding the optimum under LIP to be significantly faster than under LDP.

\section{Multiple Attributes} \label{sec:mult}

An often-occuring scenario is that a user's data consists of multiple attributes, i.e., $X = (X^1,\ldots,X^m) \in \mathcal{X} = \mathcal{X}^1 \times \cdots \times \mathcal{X}^m$. This can be problematic for our approach for two reasons:
\begin{enumerate}
\item Such a large $\mathcal{X}$ can be problematic, since the computing time for optimisation both under LDP and LIP will depend heavily on $a$.
\item In practice, an attacker might sometimes utilise side channels to access some subsets of attributes $X_i^j$ for some users. For these users, a sanitisation protocol can leak more information (w.r.t. to the attacker's updated prior information) than its LDP/LIP parameter would suggest.
\end{enumerate}

To see how the second problem might arise in practice, suppose that $X^1_i$ is the height of individual $i$, $X^2_i$ is their weight, and $S_i$ is whether $i$ is obese or not. Since height is only lightly correlated with obesity, taking $Y_i = X_i^1$ would satisfy $\varepsilon$-LIP for some reasonably small $\varepsilon$. However, suppose that an attacker has access to $X^2_i$ via a side channel. While knowing $i$'s weight gives the attacker some, but not perfect knowledge about $i$'s obesity, the combination of the weight from the side channel, and the height from the $Y_i$, allows the attacker to calculate $i$'s BMI, giving much more information about $i$'s obesity. Therefore, the given protocol gives much less privacy in the presence of this side channel.

To solve the second problem, we introduce a more stringent privacy notion called \emph{Side-channel Resistant LIP} (SRLIP), which ensures that no matter which attributes an attacker has access to, the protocol still satisfies $\varepsilon$-LIP with respect to the attacker's new prior distribution. One could similarly introduce SRLDP, and many results will still hold for this privacy measure; nevertheless, since we concluded that LIP is preferable to LDP, we focus on SRLIP. For any subset $J \subseteq \{1,\ldots,m\}$, we write $\mathcal{X}^J = \prod_{j \in J} \mathcal{X}^j$ and its elements as $x^J$.

\begin{definition} \emph{($\varepsilon$-SRLIP)}. Let $\varepsilon > 0$, and let $\mathcal{X} = \prod_{j=1}^m \mathcal{X}^j$. 
We say that $\mathcal{Q}$ satisfies $\varepsilon$-SRLIP if for every $y \in \mathcal{Y}$, for every $s \in \mathcal{S}$, for every $J \subseteq \{1,\ldots,m\}$, and for every $x^J \in \mathcal{X}^J$ one has
\begin{equation}
\textrm{\emph{e}}^{-\varepsilon} \leq \frac{\mathbb{P}(Y = y | S = s, X^J = x^J)}{\mathbb{P}(Y = y | X^J = x^J)} \leq \textrm{\emph{e}}^{\varepsilon}.
\end{equation}

\end{definition}

In terms of Remark \ref{rem:prior}, $\mathcal{Q}$ satisfies $\varepsilon$-SRLIP if and only if it satisfies $\varepsilon$-LIP w.r.t. $\recht{p}_{S,X|x^J}$ for all $J$ and $x^J$. Taking $J = \varnothing$ gives us the regular definition of $\varepsilon$-LIP, proving the following Lemma:

\begin{lemma} \label{lem:srlip}
Let $\varepsilon > 0$. If $\mathcal{Q}$ satisfies $\varepsilon$-SRLIP, then $\mathcal{Q}$ satisfies $\varepsilon$-LIP.
\end{lemma}

While SRLIP is stricter than LIP itself, it has the advantage that even when an attacker has access to some data of a user, the sanitisation protocol still does not leak an unwanted amount of information beyond the knowledge the attacker has gained via the side channel. Another advantage is that, contrary to LIP itself, SRLIP satisfies an analogon of the concept of \emph{privacy budget} \cite{dwork2006calibrating}:

\begin{theorem} \label{thm:budget} 
Let $\mathcal{X} = \prod_{j=1}^m \mathcal{X}^j$, and for every $j$, let $\mathcal{Q}^j\colon \mathcal{X}^j \rightarrow \mathcal{Y}^j$ be a sanitisation protocol. Let $\varepsilon^j \in \mathbb{R}_{\geq 0}$ for every $j$. Suppose that for every $j \leq m$, for every $J \subseteq\{1,\ldots,j-1,j+1,\ldots,m\}$, and every $x^{J} \in \mathcal{X}^{J}$, $\mathcal{Q}^j$ satisfies $\varepsilon^j$-LIP w.r.t. $\recht{p}_{S,X|x^{J}}$. Then $\prod_j \mathcal{Q}^j \colon \mathcal{X} \rightarrow \prod_j \mathcal{Y}^j$ satisfies $\sum_j \varepsilon^j$-SRLIP.
\end{theorem}

The proof is presented in Appendix \ref{app:proof}. This theorem tells us that to find a $\varepsilon$-SRLIP protocol for $\mathcal{X}$, it suffices to find a sanitisation protocol for each $\mathcal{X}^j$ that is $\frac{\varepsilon}{m}$-LIP w.r.t. a number of prior distributions. Unfortunately, the method of finding an optimal $\varepsilon$-LIP protocol w.r.t. one prior $\recht{p}_{S,X}$ of Theorem \ref{thm:lip} does not transfer to the multiple prior setting. This is because this method only finds one $(R,q)$, while by (\ref{eq:LIP2}) we need a different $(R,q)$ for each prior distribution. Therefore, we are forced to adopt an approach similar to the one in Theorem \ref{thm:ldp}. The matrix $Q^j$ (given by $Q^j_{y^j|x^j} = \mathbb{P}(\mathcal{Q}^j(x^j) = y^j$)) corresponding to $\mathcal{Q}^j\colon \mathcal{X}^j \rightarrow \mathcal{Y}^j$ satisfies the criteria of Theorem \ref{thm:budget} if and only if the following criteria are satisfied:

\begin{align}
\forall x^j: & \ \sum_{y^j} Q^j_{y^j|x^j} = 1,\label{eq:rslip1}\\
\forall x ^j,y^j: & \ 0 \leq Q^j_{y^j|x^j},\\
\forall J, x^J, s,y^j: & \ \textrm{e}^{-\varepsilon/m}(Q^j\recht{p}_{X^j|x^{J}})_{y^j} \leq (Q^j\recht{p}_{X^j|s,x^{J}})_{y^j},\\
\forall J, x^J, s,y^j: & \ (Q^j\recht{p}_{X^j|s,x^{J}})_{y^j} \leq \textrm{e}^{\varepsilon/m} (Q^j\recht{p}_{X^j|x^{J}})_{y^j}. \label{eq:rslip5}
\end{align}

Similar to Theorem \ref{thm:ldp}, we can find the optimal $\mathcal{Q}^j$ satisfying these conditions by finding the vertices of the polytope defined by (\ref{eq:rslip1}--\ref{eq:rslip5}). In terms of time complexity, the comparison to finding the optimal $\varepsilon$-LIP protocol via Theorem \ref{thm:lip} versus finding a $\varepsilon$-SRLIP protocol via Theorem \ref{thm:budget} is not straightforward. The complexity of enumerating the vertices of a polytope is $\mathcal{O}(ndv)$, where $n$ is the number of inequalities, $d$ is the dimension, and $v$ is the number of vertices \cite{avis1992pivoting}. For the $\Delta$ of Theorem \ref{thm:lip} we have $d = a-1$ and $n = a+2c$. In contrast, the polytope defined by (\ref{eq:rslip1}--\ref{eq:rslip5}) satisfies $d = a^j(a^j-1)$ and $n = (a^j)^2+2c\prod_{j' \neq j} (a^{j'}+1)$. Finding $v$ for both these polytopes is difficult, but in general $v \leq \binom{n}{d}$. Since this grows exponentially in $d$, we expect Theorem \ref{thm:budget} to be faster when the $a^j$ are small compared to $a$, i.e., when $m$ is large. We will investigate this experimentally in the next section.

\section{Explicit protocols} \label{sec:cr}

The methods of Sections \ref{sec:ldp} and \ref{sec:lip} allow us to find the optimal LDP and LIP protocols.
The complexity depends heavily on $a$ and $c$, and can become computationally infeasible for large $a$ and $c$. 
For such datasets, one has to rely on predetermined privacy algorithms. 
We consider two approaches: as a benchmark, we discuss how `standard' LDP protocols can be applied to the Privacy Funnel situation, and we introduce a new method, Conditional Reporting, that is meant to address the shortcomings of standard LDP protocols. As in the previous section, we focus on LIP, but much of the discussion carries over to LDP as well.

\subsection{Standard LDP protocols}

In the literature, there are many examples of protocols $\mathcal{Q}\colon \mathcal{X} \rightarrow \mathcal{Y}$, depending on a privacy parameter $\alpha$, whose output satisfies $\alpha$-LDP with respect to $X$; for an overview see \cite{yang2020local}. Such a protocol automatically satisfies $\alpha$-LDP, hence certainly $\alpha$-LIP, with respect to $S$. However, because $X$ is only indirectly correlated with $Y$, such a protocol's actual LIP value may be lower. We can find the privacy of such a protocol $\mathcal{Q}$ by
\begin{equation}
\recht{LIP}(\mathcal{Q}) = \max_{y \in \mathcal{Y},s \in \mathcal{S}} \left|\ln \frac{\sum_x Q_{y|x}\recht{p}_{x|s}}{\sum_x Q_{y|x}\recht{p}_x}\right|;
\end{equation}
then $\mathcal{Q}$ satisfies $\varepsilon$-LIP if and only if $\recht{LIP}(\mathcal{Q}) \leq \varepsilon$. 

For this paper we are mainly interested in two protocols. The first one is Generalised Rapid Response (GRR) \cite{warner1965randomized}. We are interested in GRR because for large enough $\alpha$ it maximises $\recht{I}(X;Y)$ \cite{kairouz2014extremal}. Given $\alpha$, GRR is a privacy protocol $\recht{GRR}^{\alpha}\colon\mathcal{X} \rightarrow \mathcal{X}$ given by
\begin{equation}
\recht{GRR}^{\alpha}_{y|x} = \left\{\begin{array}{ll}
\frac{\textrm{e}^{\alpha}}{\textrm{e}^{\alpha}+a-1}, & \textrm{ if $x = y$},\\
\frac{1}{\textrm{e}^{\alpha}+a-1}, & \textrm{ if $x \neq y$}.
\end{array}\right.
\end{equation}
A direct calculation then shows that
\begin{equation} \label{eq:lipgrr}
\recht{LIP}(\recht{GRR}^{\alpha}) = \max_{x,s} \left|\ln\frac{1+(\textrm{e}^{\alpha}-1)\recht{p}_{x|s}}{1+(\textrm{e}^{\alpha}-1)\recht{p}_{x}}\right|.
\end{equation}
If we want GRR to satisfy $\varepsilon$-LIP, we then need to solve $\recht{LIP}(\recht{GRR}^{\alpha}) = \varepsilon$ for $\alpha$. Since $\recht{LIP}(\recht{GRR}^{\alpha})$ is increasing in $\alpha$, this can be done fast computationally.

The second protocol that is relevant to this paper is Optimised Unary Encoding (OUE) \cite{wang2017locally}. This protocol is notable for being one of the protocols that has the least known variance in frequency estimation \cite{wang2017locally}. For a choice of $\alpha$ as privacy parameter, and an input $x$, the output of $\recht{OUE}^{\alpha}\colon \mathcal{X} \rightarrow 2^{\mathcal{X}}$ is a vector of independent Bernoulli variables $E_{x'}$ for $x' \in \mathcal{X}$, satisfying
\begin{equation}
\mathbb{P}(E_{x'} = 1) = \left\{\begin{array}{ll}
\frac{1}{2}, & \textrm{ if $x'= x$},\\
\frac{1}{\textrm{e}^{\alpha}+1}, & \textrm{ if $x' \neq x$}.
\end{array}\right.
\end{equation}
In other words, If we identify a $y \in 2^{\mathcal{X}}$ with a subset of $\mathcal{X}$ (so $\#y$ denotes its cardinality), we get
\begin{equation}
\recht{OUE}^{\alpha}_{y|x} = \left\{\begin{array}{ll}
     \frac{\textrm{e}^{(a-\#y)\alpha}}{2(\textrm{e}^{\alpha}+1)^{a-1}},& \textrm{ if $x \in y$},  \\
     \frac{\textrm{e}^{(a-\#y-1)\alpha}}{2(\textrm{e}^{\alpha}+1)^{a-1}},& \textrm{ if $x \notin y$}. 
\end{array}\right.
\end{equation}

It follows that
\begin{equation}
\recht{LIP}(\recht{OUE}^{\alpha}) = \max_{y,s} \left|\ln\frac{1+(\textrm{e}^{\alpha}-1)\sum_{x\in y}\recht{p}_{x|s}}{1+(\textrm{e}^{\alpha}-1)\sum_{x\in y}\recht{p}_{x}}\right|.
\end{equation}

\subsection{Conditional Reporting}

In general, a generic LDP protocol will not be ideal for our situation, since these are designed to obscure all information about $X$, rather than just the part that holds information about~$S$. To address this shortcoming, we introduce the {\em Conditional Reporting} (CR) in Algorithm \ref{alg:cr}. This mechanism needs both $S$ and $X$ as input; hence it differs from the other protocols discussed in this paper, which only have $X$ as input. The value of $S$ is masked by Randomised Response. If the output $\tilde{s}$ equals $S$, we return the true value of $X$. If not, we output a random one, whose probability distribution is given by $\recht{p}_{X|\tilde{s}}$.

\begin{algorithm}
\SetAlgoLined
\SetKwInOut{Input}{Input}\SetKwInOut{Output}{Output}
\SetKwFunction{Sort}{Sort}
\Input{Privacy parameter $\alpha$; Probability distribution $\recht{p}_{S,X}$; input $(s,x) \in \mathcal{S} \times \mathcal{X}$}
\Output{ $y \in \mathcal{X}$}
\BlankLine
Sample $\tilde{S} \in \mathcal{S}$ with $\mathbb{P}(\tilde{S} = s') = \left\{\begin{array}{ll} \frac{\textrm{e}^{\alpha}}{\textrm{e}^{\alpha}+\#\mathcal{S}-1}, & \textrm{ if $s' = s$},\\
\frac{1}{\textrm{e}^{\alpha_0}+\#\mathcal{S}-1}, & \textrm{ otherwise}\;\end{array}\right.$
 \caption{Conditional Reporting ($\recht{CR}^{\alpha}$)}
\eIf{$\tilde{s} = s$}{$y \leftarrow x$\;}{Sample $\tilde{x} \in \mathcal{X}$ with $\mathbb{P}(\tilde{x} = x') = \recht{p}_{x'|\tilde{s}}$\;
$y \leftarrow \tilde{x}$\;}
 \label{alg:cr}
\end{algorithm}

$\recht{CR}^{\alpha}$ certainly satisfies $\alpha$-LDP, hence $\alpha$-LIP, w.r.t. $S$. However, if $S$ and $X$ are not perfectly correlated, we can get better privacy, as outlined by the proposition below.

\begin{proposition} \label{prop:cr}
Given a probability distribution $\recht{p}_{X,S}$ and a $\alpha \geq 0$, define
\begin{equation} \label{eq:lipcr}
L(\alpha) = \max_{x,s} \left|\ln \frac{(\textrm{e}^{\alpha}-1)\recht{p}_{x|s}+\sum_{s'}\recht{p}_{x|s'}}{(\textrm{e}^{\alpha}-1)\recht{p}_{x}+\sum_{s'}\recht{p}_{x|s'}} \right|.
\end{equation}
Then $\recht{CR}^{\alpha}$ satisfies $\varepsilon$-LIP if and only if $\varepsilon \geq L(\alpha)$.
\end{proposition}

The proof is presented in Appendix \ref{app:proof}. One can use this proposition to find the $\alpha$ needed to have $\recht{CR}^{\alpha}$ satisfy $\varepsilon$-LDP, by solving $L(\alpha) = \varepsilon$. At the very least one has the following upper bound:

\begin{proposition}
The protocol $\recht{CR}^{\alpha}$ satisfies $\alpha$-LDP. In particular, it satisfies $\alpha$-LIP, and $L(\alpha) \leq \alpha$.
\end{proposition}

\begin{proof}
For all $y \in \mathcal{X}$ and $s \in \mathcal{S}$ we have, following equation (\ref{eq:app1}) in Appendix A, that
\begin{equation}
\mathbb{P}(\recht{CR}^{\alpha}(X,S) = y|S = s) = \tfrac{1}{\textrm{e}^{\alpha}+c-1}
\Big(\textrm{e}^{\alpha}\recht{p}_{y|s}+\sum_{s' \neq s}\recht{p}_{y|s'}\Big).
\end{equation}
It follows that
\begin{align}
&\frac{\mathbb{P}(\recht{CR}^{\alpha}(X,S) = y|S = s)}{\mathbb{P}(\recht{CR}^{\alpha}(X,S) = y|S = s')} \nonumber \\
&= \frac{\textrm{e}^{\alpha}\recht{p}_{y|s}+p_{y|s'}+\sum_{s'' \neq s,s'}\recht{p}_{y|s''}}{\recht{p}_{y|s}+\textrm{e}^{\alpha}p_{y|s'}+\sum_{s'' \neq s,s'}\recht{p}_{y|s''}} \\
&\leq \max\left\{1,\frac{\textrm{e}^{\alpha}\recht{p}_{y|s}+p_{y|s'}}{\recht{p}_{y|s}+\textrm{e}^{\alpha}p_{y|s'}}\right\} \\
&\leq \textrm{e}^{\alpha}. \qedhere
\end{align}
\end{proof}

\section{Experiments} \label{sec:exp}\vspace{-0pt}

We test the feasibility of the different methods by performing small-scale experiments on synthetic data and real-world data. All experiments are implemented in Matlab and conducted on a PC with Intel Core i7-7700HQ 2.8GHz and 32GB memory. 

\subsection{Synthetic data: LDP vs LIP}

We compare the computing time for finding optimal $\varepsilon$-LDP and $\varepsilon$-LIP protocols for $c = 2$ and $a = 5$ for 10 random distributions $\recht{p}_{S,X}$, obtained by generating each $\recht{p}_{s,x}$ uniformly from $[0,1]$ and then normalising. 
We take $\varepsilon \in \{0.5,1,1.5,2\}$; the results are in Figure \ref{fig:ldplip}. As one can see, Theorem \ref{thm:lip} gives significantly faster results than Theorem \ref{thm:ldp}; the average computing time for Theorem \ref{thm:ldp} for $\varepsilon = 0.5$ is 133s, while for Theorem \ref{thm:lip} this is 0.0206s.  With regards to the utility $\recht{I}(X;Y)$, since $\varepsilon$-LDP implies $\varepsilon$-LIP, the optimal $\varepsilon$-LIP protocol will have better utility than the optimal $\varepsilon$-LDP protocol. However, as can be seen from the figure, the difference in utility is relatively low.

Note that for bigger $\varepsilon$, both the difference in computing time and the difference in $\recht{I}(X;Y)$ between LDP and LIP become less. This is because of the probabilistic relation between $S$ and $X$, for $\varepsilon$ large enough, any sanitisation protocol satisfies $\varepsilon$-LIP and $\varepsilon$-LDP. This means that as $\varepsilon$ grows, the resulting polytopes will have fewer defining inequalities, hence they will have fewer vertices. This results in lower computation times, which affects LDP more than LIP. At the same time, the fact that every protocol is both $\varepsilon$-LIP and $\varepsilon$-LDP will result in the same optimal utility.

In Figure \ref{fig:ldplip2}, we compare optimal $\frac{\varepsilon}{2}$-LDP protocols to optimal $\varepsilon$-LIP protocols. Again, LIP is significantly faster than LDP. Since $\varepsilon$-LIP implies $\frac{\varepsilon}{2}$-LDP, the optimal $\frac{\varepsilon}{2}$-LDP has higher utility; again the difference is low.

\begin{figure}
\centering
\includegraphics[width = \linewidth]{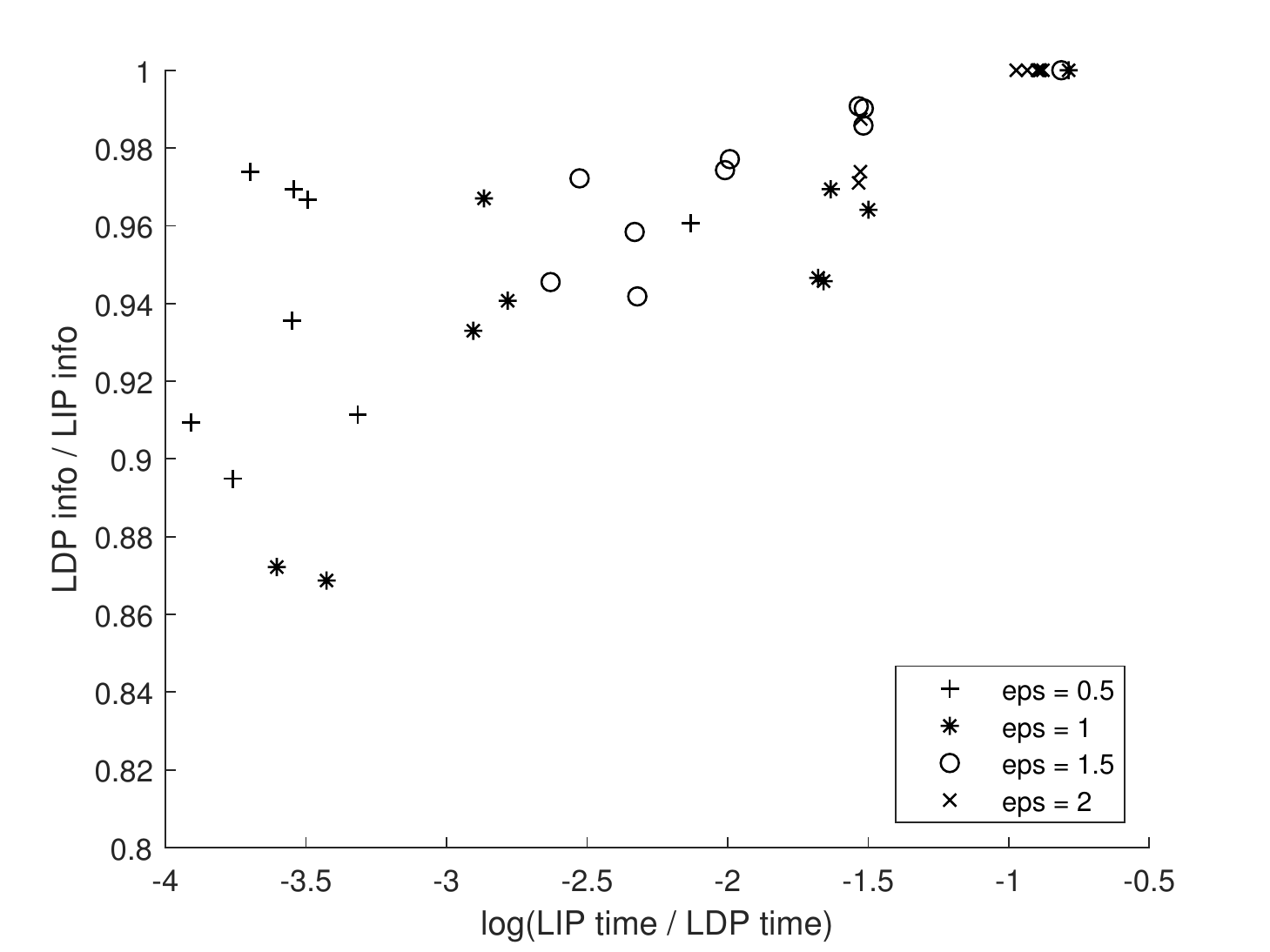}
\captionsetup{font={footnotesize,rm},justification=centering,labelsep=period}
\caption{\it Comparision of computation time and $\recht{I}(X;Y)$ for $\varepsilon$-LDP protocols found via Theorem \ref{thm:ldp} and $\varepsilon$-LIP protocols found via Theorem \ref{thm:lip}, for random $\recht{p}_{S,X}$ with $c = 2$, $a = 5$, and $\varepsilon \in \{0.5,1,1.5,2\}$.\label{fig:ldplip}}
\end{figure}

\begin{figure}
\includegraphics[width = \linewidth]{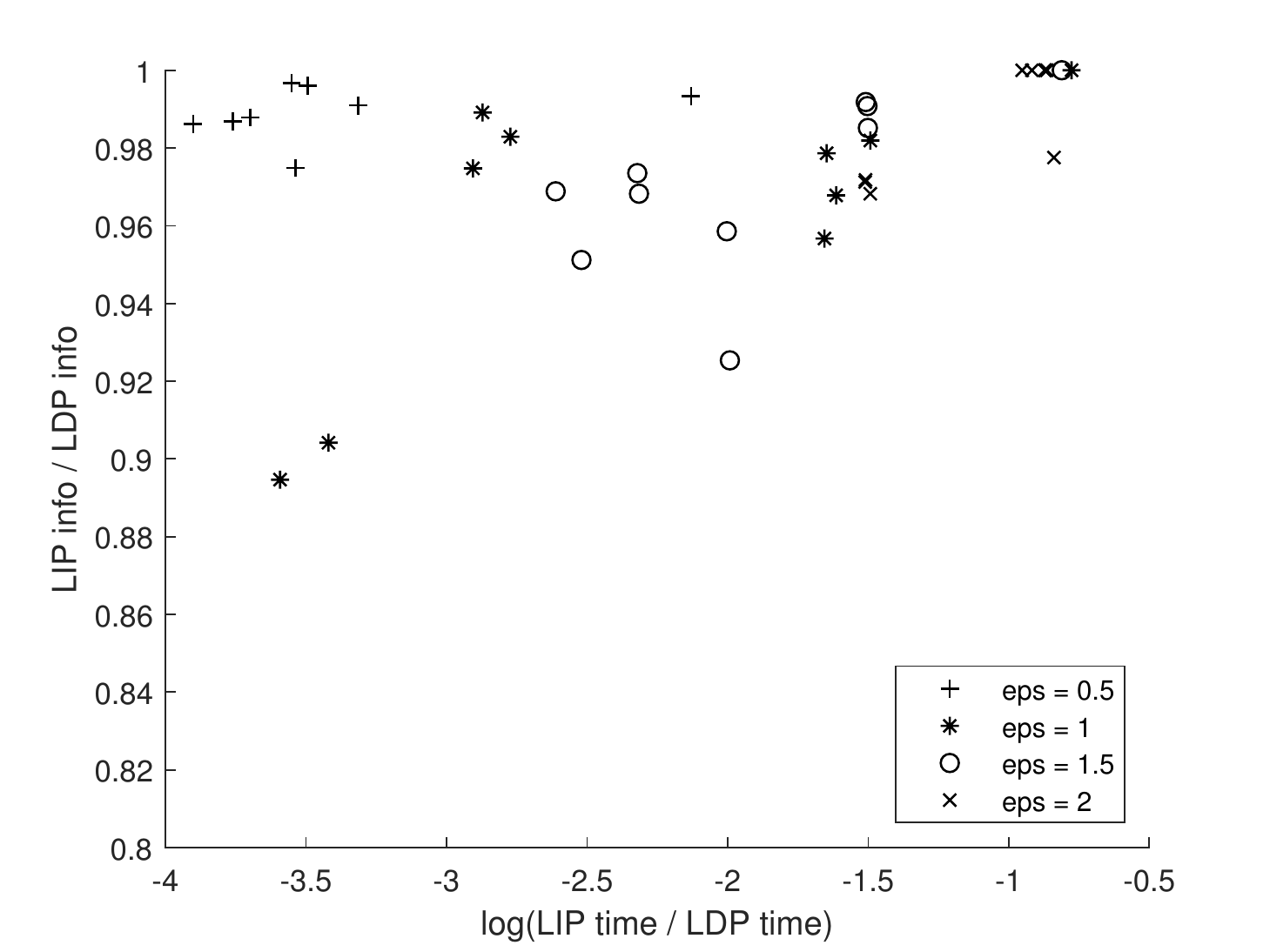}
\centering
\captionsetup{font={footnotesize,it},justification=centering,labelsep=period}
\caption{Comparision of computation time and $\recht{I}(X;Y)$ for $\varepsilon$-LDP protocols found via Theorem \ref{thm:ldp} and $\frac{\varepsilon}{2}$-LIP protocols found via Theorem \ref{thm:lip}, for random $\recht{p}_{S,X}$ with $c = 2$, $a = 5$, and $\varepsilon \in \{0.5,1,1.5,2\}$.\label{fig:ldplip2}}
\end{figure}

\subsection{Synthetic data: LIP vs SRLIP}

We also perform similar comparisons for multiple attributes, for $c = 2$, $a_1 = a_2 = 3$ and $a_3 = 4$, comparing the methods of Theorems \ref{thm:lip} and \ref{thm:budget}. The results are presented in Figure~\ref{fig:rslip}. As one can see, Theorem \ref{thm:budget} is significantly slower, with Theorem \ref{thm:lip} being on average $476$ times as fast. There is a sizable difference in utility, caused on one hand by the fact that $\varepsilon$-SRLIP is a stricter privacy requirement than $\varepsilon$-LIP, and on the other hand by the fact that Theorem \ref{thm:budget} does not give us the optimal $\varepsilon$-SRLIP protocol.

\begin{figure}
\centering
\includegraphics[width = \linewidth]{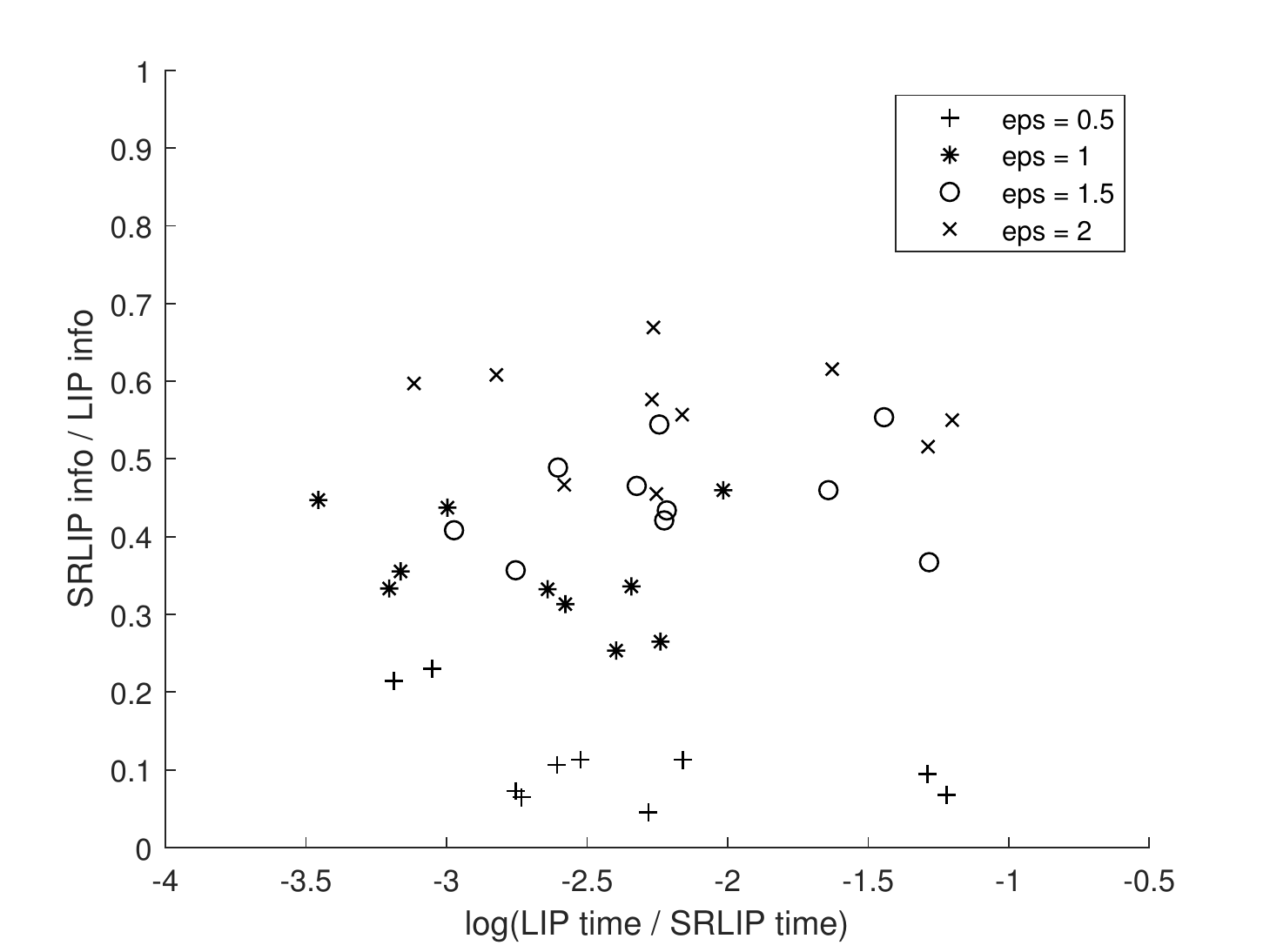}
\captionsetup{font={footnotesize,it},justification=centering,labelsep=period}
\caption{Comparison of computation time and $\recht{I}(X;Y)$ for $\varepsilon$-(SR)LIP-protocols found via Theorems \ref{thm:lip} and \ref{thm:budget}, for random $\recht{p}_{S,X}$ with $c = 2$, $a_1 = a_2 = 3$, $a_3 = 4$, and $\varepsilon \in \{0.5,1,1.5,2\}$.\label{fig:rslip}}
\end{figure}

\subsection{Adult dataset}

\begin{figure*}
    \centering
    \begin{subfigure}[b]{0.45\textwidth}
		\includegraphics[width=\textwidth]{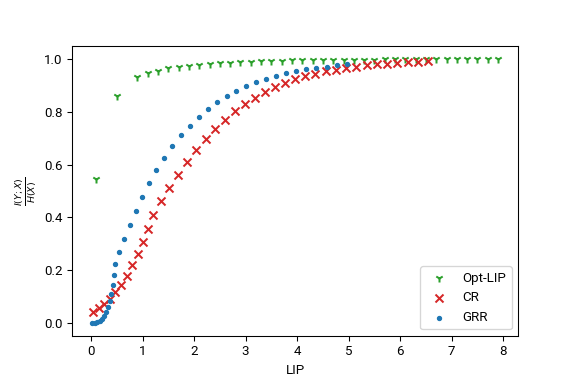}
		\vspace{-0.7cm}
		\caption{$S =$ marital status, $X =$ education}
		\label{ms_edu}
	\end{subfigure}
	\begin{subfigure}[b]{0.45\textwidth}
		\includegraphics[width=\textwidth]{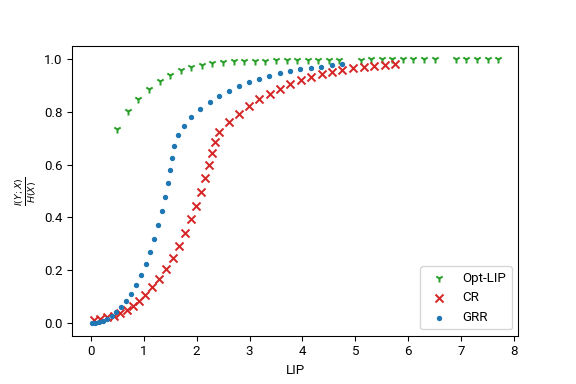}
		\vspace{-0.7cm}
		\caption{$S =$ occupation, $X =$ education}
		\label{oc_edu}
	\end{subfigure}\\
	\begin{subfigure}[b]{0.45\textwidth}
		\includegraphics[width=\textwidth]{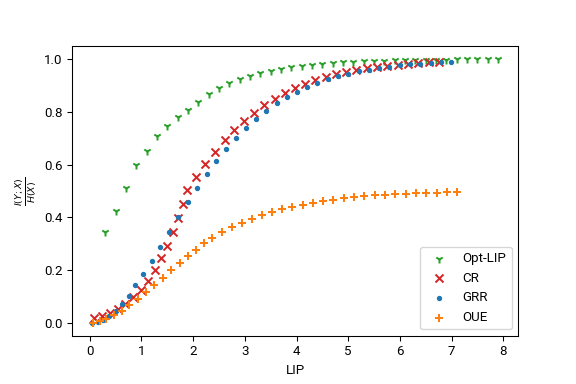}
		\vspace{-0.7cm}
		\caption{$S =$ marital status, $X =$ relationship}
		\label{ms_rel}
	\end{subfigure}
    \begin{subfigure}[b]{0.45\textwidth}
		\includegraphics[width=\textwidth]{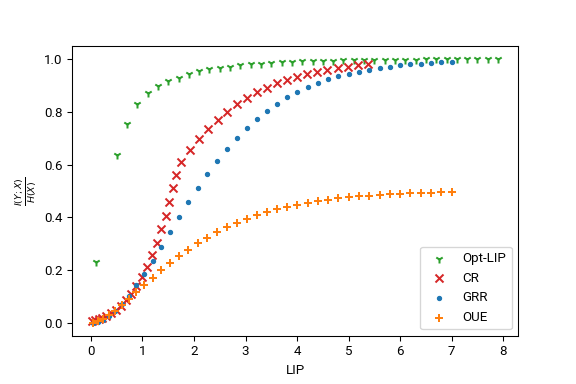}
		\vspace{-0.7cm}
		\caption{$S =$ occupation, $X =$ relationship}
		\label{oc_rel}
	\end{subfigure}\\
    \begin{subfigure}[b]{0.45\textwidth}
		\includegraphics[width=\textwidth]{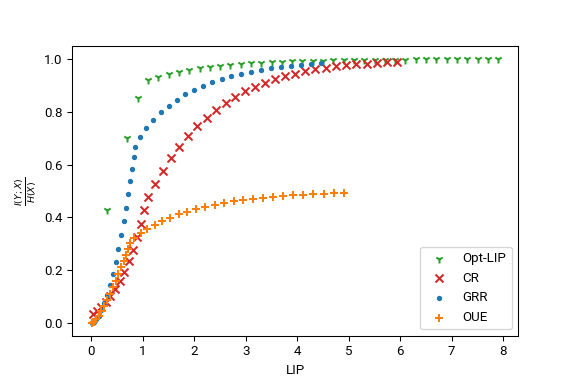}
		\vspace{-0.7cm}
		\caption{$S =$ marital status, $X =$ race}
		\label{ms_rac}
	\end{subfigure}
    \begin{subfigure}[b]{0.45\textwidth}
		\includegraphics[width=\textwidth]{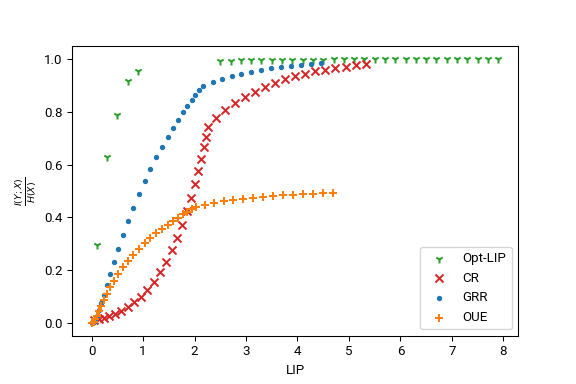}
		\vspace{-0.7cm}
		\caption{$S =$ occupation, $X =$ race}
		\label{oc_rac}
	\end{subfigure}
	\caption{\it Experiments on the Adult dataset.}
	\label{fig:adult}
\end{figure*}

We also test the utility of Conditional Reporting (CR), both on real world data and synthetic data. We consider the well-known Adult dataset \cite{Dua:2019}, which contains demographic data from the 1994 US census. 
For our tests, we take $S \in \{\text{marital status, occupation}\}$ (with $c = 7$ and $c=15$, respectively) and $X \in \{\text{education},\text{relationship},\text{sex}\}$ (with $a = 16,6,2$). 
Based on our findings in the previous sections, we take LIP as a privacy measure, and $\recht{I}(X;Y)$ as a utility measure. 
We compare CR on the one hand with the optimal method (Opt-LIP) found in Section \ref{sec:lip}, and on the other hand with the established LDP protocols GRR and OUE. The results are shown in Figure \ref{fig:adult}. For $X =$ education, the mutual information for OUE was infeasible to compute. Similarly, for $S = $ occupation, some cases of Opt-LIP failed to compute within a reasonable timeframe. Nevertheless, we can conclude that GRR and CR both perform somewhere between Opt-LIP and OUE. As the LIP value $\varepsilon$ grows larger, GRR and CR grow close to Opt-LIP. At the same time, OUE falls off for large $\varepsilon$, having $\tfrac{1}{2}\recht{H}(X)$ as its limit. This is because OUE only has probability $\tfrac{1}{2}$ transmitting the true $X$ (as element of the set $Y$). The difference between GRR and CR is less clear, and it appears to depend on the joint distribution $\recht{p}_{X,S}$ which protocol gives the best utility.

\subsection{Synthetic data: GRR vs CR}

To investigate the difference between GRR and CR, we apply both methods to synthetic data. We disregard OUE as it performs worse than the other two protocols, especially in the low privacy regime. For a fixed choice of $a$ and $c$, we draw a number of probability distributions from the Jeffreys prior on $\mathcal{S}\times\mathcal{X}$, i.e. the symmetric Dirichlet distribution with parameter $\tfrac{1}{2}$. We fix a set of LIP values $\varepsilon$, and for each  of these and each probability distribution, we solve equations (\ref{eq:lipgrr}) and (\ref{eq:lipcr}), setting the left hand side equal to $\varepsilon$ and solving for $\alpha_{\recht{GRR}}$ and $\alpha_{\recht{CR}}$. We then calculate the mutual information $\recht{I}(X;Y)$, which we normalise by dividing by $\recht{H}(X)$. The resulting averages and standard deviations are displayed in Figure \ref{fig:synth}. On the whole, we see that the larger $a$ is compared to $c$, the more utility CR provides compared to GRR. However, this does not tell the whole story, as the difference between datasets has more impact on the utility than the difference between methods.

\begin{figure*}
    \centering
    \begin{subfigure}[b]{0.45\textwidth}
		\includegraphics[width=\textwidth]{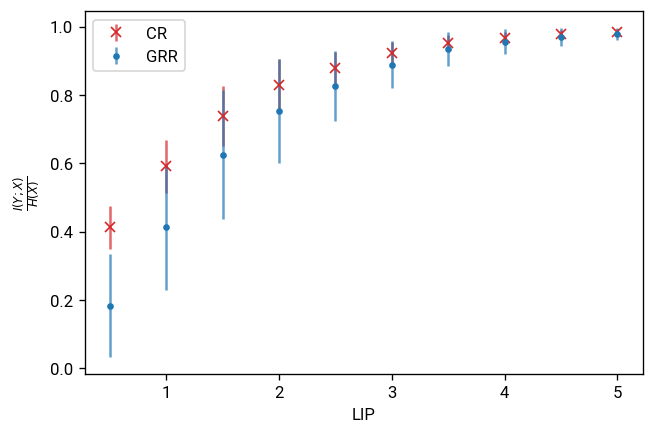}
		\vspace{-0.7cm}
		\caption{$a = 5$, $c = 2$}
		\label{a5c5}
	\end{subfigure}
	\begin{subfigure}[b]{0.45\textwidth}
		\includegraphics[width=\textwidth]{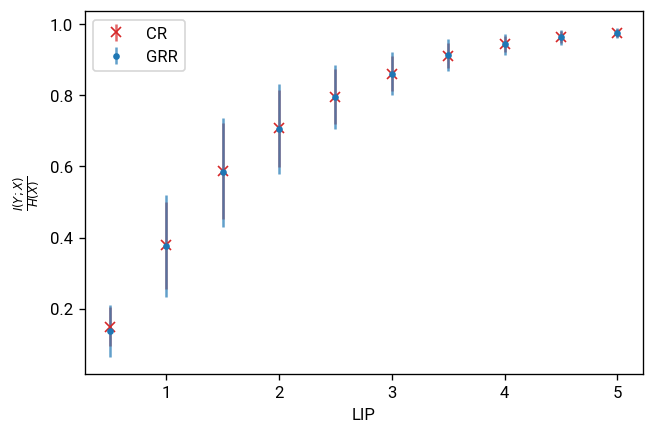}
		\vspace{-0.7cm}
		\caption{$a = 2$, $c = 5$}
		\label{a2c5}
	\end{subfigure}\\
	\begin{subfigure}[b]{0.45\textwidth}
		\includegraphics[width=\textwidth]{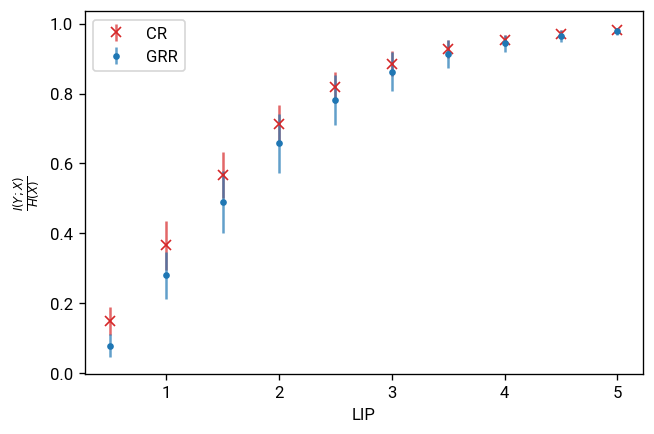}
		\vspace{-0.7cm}
		\caption{$a = 5$, $c = 5$}
		\label{a5c5}
	\end{subfigure}
	\begin{subfigure}[b]{0.45\textwidth}
		\includegraphics[width=\textwidth]{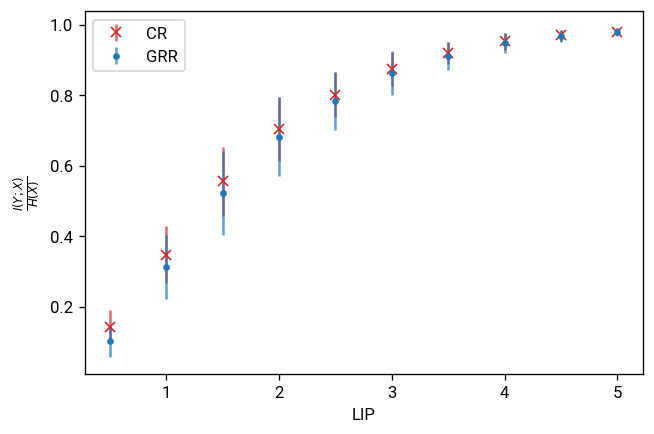}
		\vspace{-0.7cm}
		\caption{$a = 3$, $c = 5$}
		\label{a3c5}
	\end{subfigure}\\
    \begin{subfigure}[b]{0.45\textwidth}
		\includegraphics[width=\textwidth]{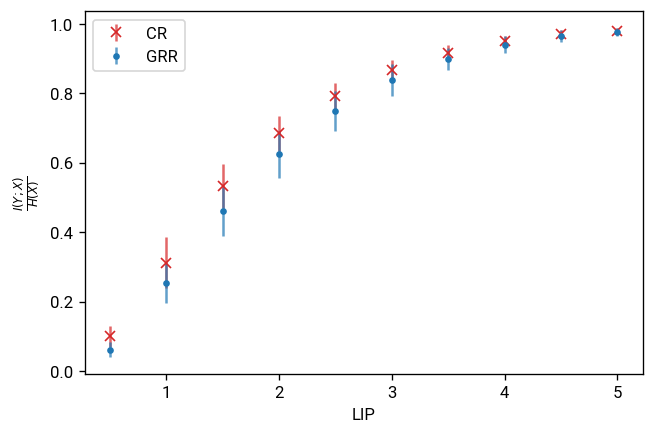}
		\vspace{-0.7cm}
		\caption{$a = 5$, $c = 7$}
		\label{a5c7}
	\end{subfigure}
    \begin{subfigure}[b]{0.45\textwidth}
		\includegraphics[width=\textwidth]{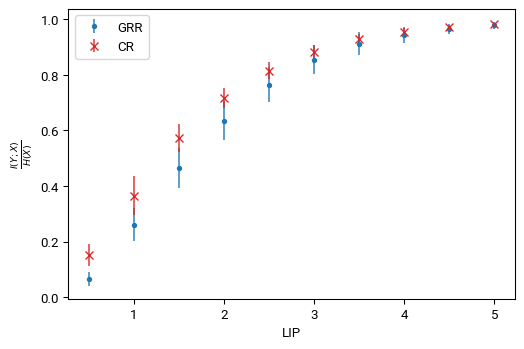}
		\vspace{-0.7cm}
		\caption{$a = 7$, $c = 5$}
		\label{a7c5}
	\end{subfigure}
	\caption{\it Experiments on synthetic data. For each value of $a$ and $c$, the average utility is taken over 100 randomly generated probability distributions. Bar size denotes standard deviation.}
	\label{fig:synth}
\end{figure*}

\subsection{GRR and CR parameter $\alpha$}

To investigate what property of the probability distribution $p_{XS}$
causes CR to outperform GRR, we consider the parameters $\alpha_{\recht{CR}}$ and $\alpha_{\recht{GRR}}$ that govern the privacy protocols CR and GRR. Both of these have the property that the higher their value, the less `random' the protocols are, resulting in a better utility. Since these $\alpha$ are found from $\varepsilon$ through different equations, the difference in utility of GRR and CR for different probability distributions may be explained by a difference in $\alpha$. We test this assertion for 100 randomly generated distributions in Figure \ref{fig:alpha}. As can be seen, the difference in mutual information can for a large part be explained by a difference in $\alpha$ ($\rho = 0.9815$, $\rho = 0.9889$, and $\rho = 0.9731$, respectively). In Figure \ref{fig:alphaeps}, we plot the relation between $\alpha$ and the LIP value $\varepsilon$ for the experiments in \ref{oc_edu} and \ref{oc_rel}. The fact that $\alpha_{\recht{GRR}} > \alpha_{\recht{CR}}$ in \ref{alpha1} corresponds to the fact that GRR outperforms CR in \ref{oc_edu}, and the opposite relation holds between \ref{alpha2} and \ref{oc_rel}.

\begin{figure}
    \centering
    \begin{subfigure}[b]{0.45\textwidth}
		\includegraphics[width=\textwidth]{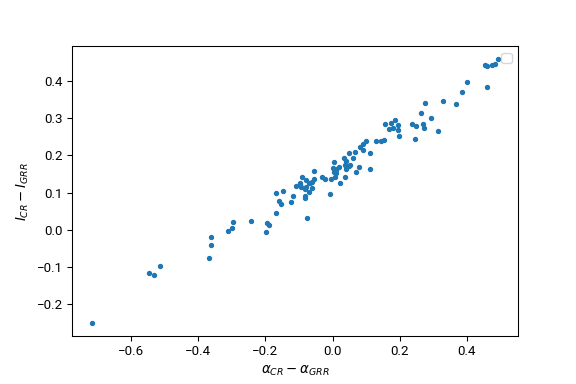}
		\vspace{-0.7cm}
		\caption{$\varepsilon = 1$}
		\label{a5c5}
	\end{subfigure}
	\begin{subfigure}[b]{0.45\textwidth}
		\includegraphics[width=\textwidth]{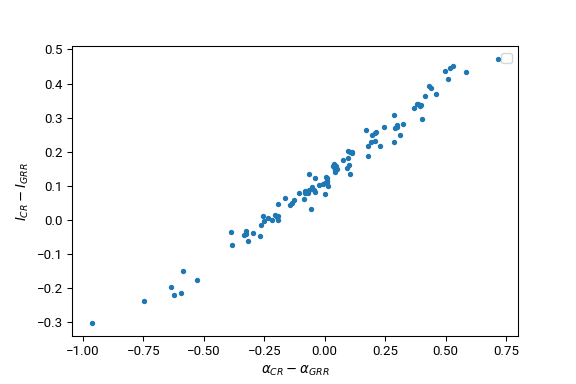}
		\vspace{-0.7cm}
		\caption{$\varepsilon = 1.5$}
		\label{a5c5}
	\end{subfigure}
    \begin{subfigure}[b]{0.45\textwidth}
		\includegraphics[width=\textwidth]{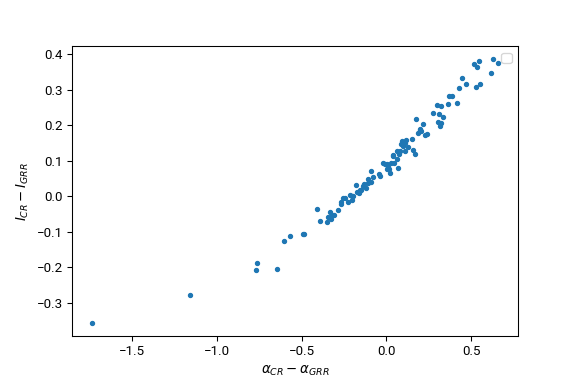}
		\vspace{-0.7cm}
		\caption{$\varepsilon = 2$}
		\label{a5c7}
	\end{subfigure}
    \caption{\it Difference in $\alpha$ versus difference in utility for $100$ randomly generated probability distributions, for $a=c=5$.}
    \label{fig:alpha}
\end{figure}

\begin{figure*}
    \centering
    \begin{subfigure}[b]{0.45\textwidth}
		\includegraphics[width=\textwidth]{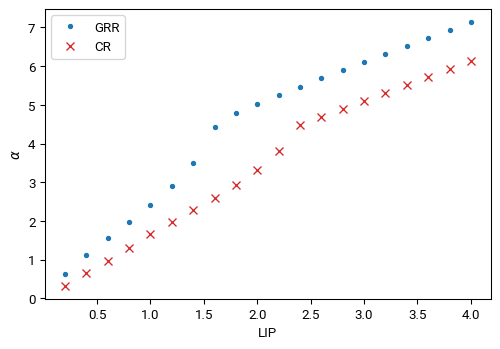}
		\vspace{-0.7cm}
		\caption{$S =$ occupation, $X =$ education}
		\label{alpha1}
	\end{subfigure}
	\begin{subfigure}[b]{0.45\textwidth}
		\includegraphics[width=\textwidth]{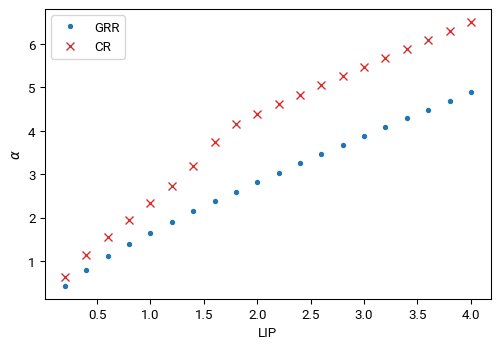}
		\vspace{-0.7cm}
		\caption{$S =$ occupation, $X =$ relationship}
		\label{alpha2}
	\end{subfigure}
    \caption{\it Value of GRR and CR parameter $\alpha$ for different values of $\varepsilon$ for the Adult dataset.}
    \label{fig:alphaeps}
\end{figure*}

Unfortunately, we were not able to relate the difference in parameter $\alpha$ to other properties of the distribution. Without presenting details we mention that the properties $\recht{I}(X;S), \max_{x,s} \recht{p}_{x,s}, \max_x \recht{p}_x$ and $\max_s \recht{p}_s$ do not appear to have an impact on the difference in utility between GRR and CR.

\section{Conclusions and future work}

Local data sanitisation protocols have the advantage of being scalable for large numbers of users. Furthermore, the advantage of using differential privacy-like privacy metrics is that they provide worst-case guarantees, ensuring that the privacy of every user is sufficiently protected. For both $\varepsilon$-LDP and $\varepsilon$-LIP we have derived methods to find optimal sanitisation protocols.  Within this setting, we have observed that $\varepsilon$-LIP has two main advantages over $\varepsilon$-LDP. First, it fits better within the privacy funnel setting, where the distribution $\recht{p}_{S,X}$ is (at least approximately) known to the estimator. Second, finding the optimal protocol is significantly faster than under LDP, especially for small $\varepsilon$. If one nevertheless prefers $\varepsilon$-LDP as a privacy metric, then it is still worthwile to find the optimal $\frac{\varepsilon}{2}$-LIP protocol, as this can be found significantly faster, at a low utility penalty.

In the multiple attributes setting, we have shown that $\varepsilon$-SRLIP provides additional privacy guarantees compared to $\varepsilon$-LIP, since without this requirement a protocol can lose all its privacy protection in the presence of side channels. Unfortunately, however, experiments show that we pay for this both in computation time and in utility.

With regard to the specific protocols, we have found that the newly introduced protocol, CR, generally outperforms OUE, especially for high values of $\varepsilon$-LIP. It behaves more or less similar to GRR, and  which of these two protocols performs best depends on properties of the joint distribution $\recht{p}_{X,S}$. In particular, it largely depends on which of the two protocols has the highest value of their governing parameter $\alpha$. Also, we have seen that CR performs better on average if $a$ is large compared to $c$.

For further research, a number of important avenues remain to be explored. First, the aggregator's knowledge about $\recht{p}_{S,X}$ may not be perfect, because they may learn about $\recht{p}_{S,X}$ through observing $(\vec{S},\vec{X})$. Incorporating this uncertainty leads to robust optimisation \cite{bertsimas2018data}, which would give stronger privacy guarantees. 

Second, it might be possible to improve the method of obtaining $\varepsilon$-SRLIP protocols via Theorem \ref{thm:budget}. Examining its proof shows that lower values of $\varepsilon^j$ may suffice to still ensure $\varepsilon$-SRLIP. Furthermore, the optimal choice of $(\varepsilon^j)_{j \leq m}$ such that $\sum_j \varepsilon^j = \varepsilon$ might not be $\varepsilon^j = \frac{\varepsilon}{m}$. However, it is computationally prohibitive to perform the vertex enumeration for many different choices of $(\varepsilon^j)_{j \leq m}$, and as such a new theoretical approach is needed to determine the optimal $(\varepsilon^j)_{j \leq m}$ from $\varepsilon$ and $\recht{p}_{S,X}$.

Third, it would be interesting to see if there are other ways to close the gap between the theoretically optimal protocol, which may be hard to compute in practice, and general LDP protocols, which do not see the difference between sensitive and non-sensitive information. This is relevant because CR needs both $S$ and $X$ as input, and there may be situations where access to $S$ is not available.

Although CR outperforms GRR and OUE for some datasets, it does not do so consistently. More research in the properties of distributions where CR fails to provide a significant advantage might lead to improved privacy protocols.

\section*{Acknowledgements}
This work was supported by NWO grant 628.001.026 (Dutch Research Council, the Hague, the Netherlands).

\printbibliography

\appendices

\section{Proofs} \label{app:proof}

\begin{proof}[Proof of Theorem \ref{thm:budget}]
For $J \subseteq\{1,\ldots, m\}$ and $j \in \{1,\ldots,m\}$, we write $J[j] := J \cap \{1,\ldots,j-1\}$. Furthermore, we write $\mathcal{X}^{\backslash J} = \prod_{j \notin J} \mathcal{X}^j$, and its elements as $x^{\backslash J}$. We write $\varepsilon := \sum_j \varepsilon^j$. We then have
\begin{align}
\recht{p}_{y|s,x^J} &= \sum_{x^{\backslash J}} \recht{p}_{y|x}\recht{p}_{x^{\backslash J}|s,x^J} \\
&= \recht{p}_{y^J|x^J}\sum_{x^{\backslash j}} \left(\prod_{j \notin J} \recht{p}_{y^j|x^j}\right)\recht{p}_{x^{\backslash J}|s,x^J} \\
&= \recht{p}_{y^J|x^J}\sum_{x^{\backslash j}} \prod_{j \notin J} \recht{p}_{y^j|x^j}\recht{p}_{x^j|s,x^{J[j]}} \\
&= \recht{p}_{y^J|x^J}\prod_{j \notin J} \sum_{x^j} \recht{p}_{y^j|x^j}\recht{p}_{x^j|s,x^{J[j]}} \\
&= \recht{p}_{y^J|x^J}\prod_{j \notin J}\recht{p}_{y^j|s,x^{J[j]}} \\
&\leq \recht{p}_{y^J|x^J}\prod_{j \notin J}\textrm{e}^{\varepsilon^j}\recht{p}_{y^j|x^{J[j]}} \\
&\leq \textrm{e}^{\varepsilon}\recht{p}_{y^J|x^J}\prod_{j \notin J}\recht{p}_{y^j|x^{J[j]}} \\
&= \textrm{e}^{\varepsilon} \recht{p}_{y|x^J}.
\end{align}
The fact that $\textrm{e}^{-\varepsilon} \recht{p}_{y|x^J} \leq \recht{p}_{y|s,x^J}$ is proven analogously.
\end{proof}

\begin{proof}[Proof of Proposition \ref{prop:cr}]
Write $Q_{y|x,s} = \mathbb{P}(\recht{CR}^{\alpha}(x,s) = y)$. Then
\begin{align}
Q_{y|x,s} &= \sum_{s'} \mathbb{P}(\recht{CR}^{\alpha}(x,s) = y | \tilde{s} = s')\mathbb{P}(\tilde{s} = s'|S = s) \\
&=\frac{\textrm{e}^{\alpha}}{\textrm{e}^{\alpha}+c-1}+\frac{1}{\textrm{e}^{\alpha}+c-1}\sum_{s' \neq s}\recht{p}_{y|s'},
\end{align}
where $\delta_{x=y}$ is the Kronecker delta. It follows that
\begin{align}
&\mathbb{P}(\recht{CR}^{\alpha}(X,S) = y|S = s) \nonumber \\
&= \sum_{x} Q_{y|x,s}\recht{p}_{x|s} \\
&= \frac{\textrm{e}^{\alpha}}{\textrm{e}^{\alpha}+c-1}\recht{p}_{y|s}+\frac{1}{\textrm{e}^{\alpha}+c-1}\sum_{s' \neq s}\recht{p}_{y|s'} \label{eq:app1}\\
&= \frac{\textrm{e}^{\alpha}-1}{\textrm{e}^{\alpha}+c-1}\recht{p}_{y|s}+\frac{1}{\textrm{e}^{\alpha}+c-1}\sum_{s'}\recht{p}_{y|s'}, \\
&\mathbb{P}(\recht{CR}^{\alpha}(X,S) = y) \nonumber \\
&= \sum_{s} \mathbb{P}(\recht{CR}^{\alpha}(X,S) = y|S = s)\recht{p}_{s} \\
&= \frac{\textrm{e}^{\alpha}}{\textrm{e}^{\alpha}+c-1}\recht{p}_{y}+\frac{1}{\textrm{e}^{\alpha}+c-1}\sum_s\sum_{s' \neq s}\recht{p}_{y|s'}\recht{p}_s\\
&= \frac{\textrm{e}^{\alpha}}{\textrm{e}^{\alpha}+c-1}\recht{p}_{y}+\frac{1}{\textrm{e}^{\alpha}+c-1}\sum_{s'}\recht{p}_{y|s'}\sum_{s \neq s'}\recht{p}_s\\
&= \frac{\textrm{e}^{\alpha}}{\textrm{e}^{\alpha}+c-1}\recht{p}_{y}+\frac{1}{\textrm{e}^{\alpha}+c-1}\sum_{s'}(\recht{p}_{y|s'}-\recht{p}_{y,s'})\\
&=\frac{\textrm{e}^{\alpha}-1}{\textrm{e}^{\alpha}+c-1}\recht{p}_{y}+\frac{1}{\textrm{e}^{\alpha}+c-1}\sum_{s'}\recht{p}_{y|s'}.
\end{align}
We find that
\begin{equation}
L(\alpha) = \max_{y,s}\left|\ln\frac{\mathbb{P}(\recht{CR}^{\alpha}(X,S) = y|S = s)}{\mathbb{P}(\recht{CR}^{\alpha}(X,S) = y)}\right|,
\end{equation}
hence $\recht{CR}^{\alpha}$ satisfies $\varepsilon$-LIP if and only if $\varepsilon \geq L(\alpha)$.
\end{proof}
\end{document}